\documentclass{llncs}
\title{An Event Structure Model for Probabilistic Concurrent Kleene Algebra}
\author{Annabelle McIver\inst{1}, Tahiry Rabehaja\inst{1,2} and Georg Struth\inst{2}}
\institute{Department of Computing\thanks{This research has been supported by the Australia Research Council Discovery Grant DP1092464 and the iMQRS Grant
from Macquarie University.}\\
Macquarie University, Sydney, Australia \\
\email{\{annabelle.mciver,tahiry.rabehaja\}@mq.edu.au}
\and Department of Computer Science \\
University of Sheffield,
United Kingdom \\
\email{g.struth@dcs.shef.ac.uk}
}

\makeatletter
\newtheorem{rep@theorem}{\rep@title}
\newcommand{\newreptheorem}[2]{%
\newenvironment{rep#1}[1]{%
 \def\rep@title{#2 \ref{##1}}%
 \begin{rep@theorem}}%
 {\end{rep@theorem}}}
\makeatother

\usepackage{amssymb}
\usepackage{amsmath}
\usepackage{xy}
\xyoption{all}

\newcommand{\pBES}{\mathbf{pBES}}

\newcommand{\BES}{\mathbf{BES}}
\newcommand{\C}{\mathcal{C}}
\renewcommand{\L}{\mathcal{L}}
\newcommand{\D}{\mathbb{D}}
\newcommand{\N}{\mathbb{N}}
\newcommand{\TT}{\mathcal{T}}
\newcommand{\EE}{\mathcal{E}}
\newcommand{\FF}{\mathcal{F}}
\newcommand{\G}{\mathcal{G}}
\newcommand{\PP}{\mathcal{P}}
\newcommand{\cfl}{\mathbf{cfl}}
\newcommand{\refby}{\sqsubseteq}
\newcommand{\prefix}{\trianglelefteq}
\newcommand{\init}{\mathbf{in}}

\renewcommand{\>}{\rangle}

\newcommand{\pc}[1]{{\ \oplus_{\!#1}\ }}
\newcommand{\supp}{\mathrm{supp}}

\newcommand{\ov}[1]{\overline{#1}}

\newcommand{\exit}{\Phi}

\begin{document}
\maketitle
\begin{abstract}
We give a new true-concurrent model for probabilistic concurrent Kleene algebra. The model is based on probabilistic event structures, which combines ideas from Katoen's work on probabilistic concurrency and Varacca's probabilistic prime event structures. The event structures are compared with a true-concurrent version of Segala's probabilistic simulation. Finally, the algebraic properties of the model are summarised to the extent that they can be used to derive techniques such as probabilistic rely/guarantee inference rules.
\end{abstract}

\section{Introduction}

The use of probability in concurrent systems has provided solutions to many problems where non-probabilistic techniques would fail~\cite{Rab76}. However, the combination of probability and concurrency increases the complexity of any formal tool powerful enough to ensure the correctness of a system involving both features. It is then imperative that such a framework should be as simple as possible and the use of algebras in formal verifications is indeed a step in that direction. In this paper, we follow an algebraic approach in the style of Hoare et al's concurrent Kleene algebra (CKA) that is sound under a true-concurrent interpretation~\cite{Hoa09}. The algebraic laws model the interactions between probability, nondeterminism, concurrency and finite iteration operators. The structure produces an algebra which is an important mathematical tool for carrying out complex verification tasks and can be used to give robust proofs of concurrent systems, and in particular for verification techniques such as Jones’ rely/guarantee rules~\cite{Hoa09,Jon12}. 

We have previously developed an interleaving model for probabilistic concurrent Kleene algebra (pCKA) that aims to combine probability and concurrency in a single algebraic setting~\cite{Rab13}.  Starting from the same set of axioms, we present a novel true-concurrent model based on bundle event structures (BES)~\cite{Kat96,Lan92}. Our motivation is that the concurrency operator of event structures provides a more faithful interpretation of concurrency found in physical systems. In contrast, the parallel composition of automata fails to capture some fundamental properties such as refinement of actions~\cite{Gor97}. Indeed, we show that our semantics distinguishes processes that are equal in the interleaving case. Event structures were introduced by Winskel~\cite{Win86} and have been studied extensively by others~\cite{Kat96,Lan92,Gla03,Gla09}, refined to bundle event structures by Langerak~\cite{Lan92} and extended to account for probabilistic specifications by Katoen~\cite{Kat96}. Katoen concentrated on event structures for probabilistic process algebras but did not provide the framework needed to compare different event structures. In contrast, Varacca studied the semantics of probabilistic prime event structures (pPES) using valuations on the set of configurations~\cite{Var03}. It is well known that prime event structures are not rich enough to express  the right factorisation of sequential composition through nondeterminism. Our true-concurrent model for pCKA requires a bundle event structure framework extended with  probabilistic simulations over the ``configuration-trees".

Our main contribution is the development of a new model for pCKA endowed with a true-concurrent version of Segala's probabilistic simulation~\cite{Seg95}. To the best of our knowledge, this is the first extension of probabilistic  simulation to the true-concurrent setting though non probabilistic versions do exist in the literature~\cite{Che92,Maj98}. We also define an adequate weakening of Katoen's techniques for pBES so that they reduce to Varacca's definitions for PES.

The paper is organised as follows. In Section~\ref{background}, we provide the necessary background for bundle event structures. The algebraic operators are defined in Section~\ref{operations} where a particular care is needed for the construction of the binary Kleene star. Without probability, we argue that bundle event structures endowed with these operators and quotiented with the pomset language equivalence forms a concrete model for CKA. In Section~\ref{pBES}, we set out the necessary tools for constructing pBES. In Section~\ref{simulation}, we define the notion of probabilistic simulation on pBES. Section~\ref{pcka} is devoted to showing that the set of pBES endowed with the defined algebraic operators modulo probabilistic simulation satisfies the axioms of pCKA. All incomplete proofs are given in full in the appendix.

\section{Bundle Event Structures}\label{background}

Event structures provide a truly concurrent denotation for processes where an event is labelled by an action from a set $\Sigma$. An event $e$ may enable another event $f$, that is, $f$ cannot happen unless $e$ has already happened. This relation, denoted by $\mapsto$, is useful for sequential dependency. It is also possible that two events cannot happen simultaneously in a single run which usually occurs when there is a nondeterministic choice of events. This second relation is denoted by $\#$ and is extended to sets of events $x,y\subseteq E$ such that $x\#y$ iff for all $e\in x$ and $f\in y$, if $e\neq f$ then $e\#f$. Formally, we have the following definition.
\begin{definition}[\cite{Lan92}]\label{def:bes}
A bundle event structure $\EE$ is a tuple $(E,\#,\mapsto,\lambda,\exit)$ such that $E$ is a set of events, $\#\subseteq E\times E$ is an irreflexive and symmetric binary relation (the conflict relation), $\mapsto\subseteq\PP(E)\times E$ is called a bundle relation where 
$$\forall x\subseteq E\ \forall e\in E: x\mapsto e\Rightarrow x\#x,$$
$\lambda:E\to\Sigma$ is a labelling (partial) function and $\Phi\subseteq E$ is a set of events such that $\Phi\#\Phi$. Elements of $\Phi$ are called final events and $\PP(E)$ is the powerset of $E$.
\end{definition}

In the bundle $x\mapsto e$, $x$ is referred to as a bundle set and the event $e$ is pointed by $x$. Since $x\#x$ holds for every $x$ such that $x\mapsto e$, it follows that exactly one event in $x$ must enable $e$ and such a unique event is required for each bundle set pointing to $e$ before it can happen. Given a set of events $x\subseteq E$, we denote by $\cfl(x) = \{e\in E\ |\ \exists e'\in x:e\#e'\}$ the set of events that are in conflict with some event in $x$. A set $x$ is called \textit{conflict free} if $\cfl(x)\cap x = \emptyset$. Unlabelled events happen without any noticeable internal nor external observable outputs. They are only used as ``delimiters".

A (finite) sequence of events $e_1e_2\cdots e_n$ from $E$ is called an \emph{event trace} if for every $i\geq 1$ and every bundle relation $y\mapsto e_i$, there exists $j<i$ such that $e_j\in y$ and $e_i\notin\cfl(\{e_1,\dots,e_{i-1}\})\cup\{e_1,\dots,e_{i-1}\}$.

\begin{definition}[\cite{Lan92}]\label{def:configuration}
A configuration is a subset $x\subseteq E$ such that $x=\{e_1,\dots,e_n\}$ for 
some event trace $e_1\cdots e_n$ referred to as a linearisation of $x$. The set of all configurations (reps. traces) of $\EE$ is denoted by $\C(\EE)$ (resp. $\TT(\EE)$).
\end{definition}

In the sequel we will need to describe the causal dependencies between events in more detail. To do this we associate a partial order with each configuration. 

A \emph{labelled partial order} (lposet) is a tuple $(x,\preceq,\lambda)$ where $(x,\preceq)$ is a poset and $\lambda:x\to\Sigma$. Unlabelled events of a lposet $u = (x,\preceq,\lambda)$ can be removed to obtain the sub-lposet $\hat u = (\hat{x},\preceq\!\!|_{\hat x},\lambda|_{\hat x})$ 
such that $\hat{x} = \{e\in x\ |\ \lambda(e) \textrm{ is defined}\}$ and where $\preceq\!\!|_{\hat x}$ and $\lambda|_{\hat x}$ are the respective restrictions of $\preceq$ and $\lambda$ to the set $\hat x$. 
A lposet $u = (x,\preceq_x,\lambda_x)$ \textit{implements} another lposet $v = (y,\preceq_y,\lambda_y)$ if there exists a label-preserving monotonic bijection $f:\hat{y}\to\hat{x}$ and we write $u\refby_s v$ or simply $x\refby_s y$ if no confusion arises ($s$ stands for subsumption~\cite{Gis88}).

Given an event trace $e_1\cdots e_n$ of a BES $\EE$, we denote by $\preceq_{e_1\cdots e_n}$ the reflexive transitive closure of the order  $\preceq$ of events in that sequence i.e. $e_1\preceq e_2, e_2\preceq e_3,\dots, e_{n-1}\preceq e_n$. The tuple $(\{e_1,\dots,e_n\},\preceq_{e_1\cdots e_n},\lambda|_{\{e_1,\dots,e_n\}})$ is a lposet. Let $x\in\C(\EE)$. We generate a lposet $(x,\preceq,\lambda)$ where $$\preceq=\bigcap_{x=\{e_1,\dots,e_n\}\wedge e_1\cdots e_n\in\TT(\EE)}\preceq_{e_1\cdots e_n}$$ and $\lambda$ is restricted to $x$. Intuitively, two events are incomparable iff neither has to happen before the other.

The set of lposets of $\EE$  is denoted $\L(\EE)$, that is, $\L(\EE) = \{(x,\preceq,\lambda)\ |\ x\in\C(\EE)\}$. Given two bundle event structures $\EE$ and $\FF$, it is well known that $\C(\EE) = \C(\FF)$ iff $\TT(\EE) = \TT(\FF)$ iff $\L(\EE) = \L(\FF)$~\cite{Kat96,Lan92}. We say that $(x,\preceq_x,\lambda_x)$ is a \emph{prefix} of $(y,\preceq_y,\lambda_y)$, written $(x,\preceq_x,\lambda_x)\prefix(y,\preceq_y,\lambda_y)$, if $x\subseteq y$ and $\lambda_y|_x = \lambda_x$ and $e\preceq_y e' \wedge e'\in x\Rightarrow e\in x\wedge e\preceq_x e'$. The next proposition shows that configurations inclusion characterises prefixing.

\begin{proposition}\label{pro:configuration-prefix}
Let $\EE$ be a BES. If $x,y\in\C(\EE)$ and $x\subseteq y$ then $(x,\preceq_x,\lambda_x)\prefix (y,\preceq_y,\lambda_y)$.
\end{proposition}

\section{Basic Operations on Bundle Event Structures}\label{operations}

A concurrent quantale is a particular kind of concurrent Kleene  algebra~\cite{Hoa09}. It is composed of two quantales that interact via the interchange law~(\ref{eq:interchange-law}). In this section, we show that the set $\BES$  of bundle event structures endowed with the following operators and partial order forms a concurrent quantale.  This model is extended to capture probability in Section~\ref{pBES}. 
\subsubsection*{Basic BES:} we start by defining the basic BES corresponding to Deadlock, Skip and one step action.
\begin{itemize}
\item Deadlock is denoted by $0$ and is associated with the BES $(\emptyset,\emptyset,\emptyset,\emptyset,\emptyset)$.
\item Skip is denoted by $1$ and is associated with $(\{e\},\emptyset,\emptyset,\emptyset,\{e\})$.
\item Each $a\in\Sigma$ is associated with $(\{e_a\},\emptyset,\emptyset,\lambda(e_a) = a,\{e_a\})$, denoted by $a$.
\end{itemize}

We fix two BES  $\EE = (E,\#_\EE,\mapsto_\EE,\lambda_\EE,\Phi_\EE)$ and $\FF= (F,\#_\FF,\mapsto_\FF,\lambda_\FF,\Phi_\FF)$ such that $E\cap F = \emptyset$. This ensures that the disjoint union of two labelling functions is again a function. We define the set $\init(\EE)\subseteq E$ such that $e\in\init(\EE)$ iff there is no $x\subseteq E$ such that $x\mapsto e$. Events in $\init(\EE)$ are called initial events.

\subsubsection*{Concurrency, sequential composition and nondeterminism~\cite{Kat96}} are defined in Fig.~\ref{fig:operators}. The concurrent composition $\EE\|\FF$ is the disjoint union of $\EE$ and $\FF$ delimited by fresh ineffectual events. Notice there is no synchronisation in $\|$, this is because we are mainly interested in lock-free concurrencies in the style of~\cite{Hoa09,Jon12,Jon81,Din02}. A special event can however be introduced to force synchronisation~\cite{Kat96,Gor97} and most of the algebraic laws remain valid. For the sequential composition, new bundles of the form $\exit_\EE\mapsto e$ for every $e\in\init(\FF)$ are added to make sure that all events of $\EE$ precede all events of $\FF$. For nondeterminism, the property $\init(\EE)\#\init(\FF)$ is imposed so that the occurrence of any initial event of $\EE$ will block every events of $\FF$ from happening (and symmetrically). The choice is resolved as soon as one event from $\EE$ or $\FF$ happens.

\begin{figure}[!ht]
\hspace{-2.5mm}\begin{minipage}{.5\linewidth}
\subsubsection{Concurrency} $\EE\|\FF$:\\

\begin{itemize}
\item set of events: $E\cup F\cup\{e,f\}$,
\item conflicts: $\#_{\EE}\cup\#_\FF$,
\item bundles: $\mapsto_\EE\cup\mapsto_\FF\cup\{\{e\}\mapsto e'\ |\ e'\in\init(\EE)\cup\init(\FF)\}\cup\{\exit_\EE\mapsto f,\exit_\FF\mapsto f\}$,
\item labelling: $\lambda\cup\lambda'$,
\item final events: $\exit_{\EE\|\FF} = \{f\}$.
 \end{itemize} 
 where $e,f\notin E\cup F$.
\end{minipage}\hspace{5mm}
\begin{minipage}{.5\linewidth}
\subsubsection{Sequential composition} $\EE\cdot\FF$ :\\

\begin{itemize}
\item set of events: $E\cup F$,
\item conflicts: $\#_{\EE}\cup\#_\FF$,
\item bundles: $\mapsto_\EE\cup\mapsto_\FF\cup\{\exit_\EE\mapsto e\ |\ e\in\init(\FF)\}$,
\item labelling : $\lambda\cup\lambda'$,
\item final events: $\exit_{\EE\cdot\FF} = \exit_{\FF}$.
\end{itemize}
\end{minipage}

\subsubsection{Nondeterminism} $\EE+\FF$ :
\begin{itemize}
\item set of events: $E\cup F$,
\item conflicts: $\#_{\EE}\cup\#_\FF\cup\mathrm{sym}(\init(\EE)\times\init(\FF))\cup\mathrm{sym}(\exit_\EE\times\exit_\FF)$,
\item bundles: $\mapsto_\EE\cup\mapsto_\FF$,
\item labelling: $\lambda\cup\lambda'$,
\item final events: $\exit_{\EE+\FF} = \exit_\EE\cup\exit_\FF$.
\end{itemize}
where $\mathrm{sym}(x\times y) = (x\times y)\cup (y\times x)$ is the symmetric closure.
\caption{Definitions of $\EE\|\FF$, $\EE\cdot\FF$ and $\EE+\FF$.}\label{fig:operators}
\end{figure}

\subsubsection{The Kleene star} is defined by constructing a complete partial order on the set of BES. We define the order $\EE\prefix\FF$, which is the sub-BES relation, such that 

\begin{eqnarray}
E&\subseteq& F\label{eq:event-inclusion}\nonumber\\
\#_\EE& = &\#_\FF\cap (E\times E)\label{eq:conflic-restriction}\nonumber\\
\mapsto_\EE& \subseteq &\mapsto_\FF\label{eq:bundle-inclusion}\nonumber\\
x\mapsto_\FF e\wedge e\in E&\Rightarrow &x\subseteq E\wedge x\mapsto_\EE e\label{eq:bundle-restriction}\nonumber\\
\lambda_\EE & = &\lambda_\FF|_{E}\nonumber\label{eq:label-restriction}\\
\exit_\EE & = &\exit_\FF\cap E\nonumber\label{eq:final-restriction}
\end{eqnarray}

We use the following binding precedence: $* , \cdot , \|, +$. The probabilistic choice $\!\pc{\alpha}\!$ (defined later) and $+$ are unordered and are parsed using brackets.


\begin{proposition}\label{pro:omega-completeness}
$(\BES,\prefix)$ is an $\omega$-complete partially ordered set, that is, any countable ascending chain has a least upper bound in $\BES$.
\end{proposition}

\begin{proof}[Sketch]
The proof that $\prefix$ is a partial order amounts to checking reflexivity, antisymmetry and transitivity which is clear. As for $\omega$-completeness, given a countable increasing sequence of BES $\EE_0\prefix\EE_1\prefix\EE_2\prefix\cdots$, we construct a BES $\EE = (\cup_i E_i,\cup_i\#_i,\cup_i\mapsto_i,\cup_i\lambda_i,\cup\exit_i)$. We can show that $\EE$ is indeed the least upper bound w.r.t $\prefix$ of the countable sequence $(\EE_i)_i$.\qed
\end{proof}

Let $\EE,\FF$ be two BES. The Kleene product of $\EE$ by $\FF$, denoted by $\EE*\FF$, is the limit of the $\prefix$-increasing sequence of BES $$\FF\prefix \FF+\EE\cdot\FF\prefix \FF+\EE\cdot(\FF + \EE\cdot\FF)\prefix\cdots$$
where adequate events renaming are needed to ensure that the sequence of BES are syntactically similar (see Fig.~\ref{fig:kleene-series} for a concrete example). Equivalently, $\EE*\FF$ is the least fixed point of $\lambda X.\FF + \EE\cdot X$ in $(\BES,\prefix)$.
\begin{figure}
\begin{tiny}$$
\hspace{-12mm}
\xymatrix{&f_0}
\hspace{7mm}\prefix\hspace{5mm}
\xymatrix{
f_0\ar@{}[r]|\#&e_0\ar@{|->}[d] \\
&f_1 }
\hspace{7mm}\prefix\hspace{7mm}
\xymatrix{
f_0\ar@{}[r]|\#&e_0\ar@{|->}[d]\ar@{|->}[dr]& \\
&f_1\ar@{}|\#[r]& e_1\ar[d] \\
&&f_2}
$$
\end{tiny}
An arrow $\mapsto$ denotes a bundle relation and $\#$ is the conflict relation. The events $f_i$ are labelled by $b$ while the $e_i$s are labelled by $a$.
\caption{The first three terms in the construction of $a*b$.}\label{fig:kleene-series}
\end{figure} 
The unary Kleene star is obtained as usual by $\EE^* = \EE*1$. The main reason behind the use of the binary Kleene star~\cite{Fok94} is that the unary version introduces unwanted sequential compositions. For instance, in normal Kleene algebras, a while loop with body $\EE$ is encoded as $(e_g\cdot \EE)^*\cdot e_{\neg g}$ where $e_{g}$ (resp. $e_{\neg g}$) is the event associated with the guard. Hence by the interchange law~(\ref{eq:interchange-law}), $((e_g\cdot \EE)^*\cdot e_{\neg g}) \| a$ can behave as $(e_g\cdot \EE)^*\cdot a\cdot e_{\neg g}$ but we would assume that each $e_g$ and the corresponding $e_{\neg g}$ are checked simultaneously. Hence, we interpret a while loop as $(e_g\cdot\EE)*e_{\neg g}$.

For convenience, we denote each component of the above sequence by $\EE*_{\le 0}\FF = \FF$, $\EE*_{\le 1}\FF = \FF + \EE\cdot\FF$, $\EE*_{\le 2}\FF = \FF + \EE\cdot(\FF + \EE\cdot\FF)$,\dots. The following proposition ensures that these operators are well defined.

\begin{proposition}\label{pro:well-defined}
Let $\EE,\FF$ be BES. Then for every $\bullet\in\{+,\cdot,\|,*\}$  $\exit_{\EE\bullet\FF}\#\exit_{\EE\bullet\FF}$.
\end{proposition}

\begin{proof}
We have $\exit_{\EE+\FF} = \exit_{\EE}\cup\exit_{\FF}$ and since $\exit_{\EE}\times\exit_{\FF}\subseteq\#_{\EE+\FF}$, it follows that $\exit_{\EE+\FF}\#_{\EE+\FF}\exit_{\EE+\FF}$. The result is clear for the case of $\EE\cdot\FF$ and $\EE\|\FF$ because $\exit_{\EE\cdot\FF}= \exit_\FF$ and $\exit_{\EE\|\FF} = \{f\}$ where $f$ is the fresh final event in the construction of $\EE\|\FF$. For the Kleene star, we have $\exit_{\EE*\FF} = \cup_i\exit_{\EE*_{\le i}\FF}$ (increasing union). Therefore, any pair of events  $(e,e')\in\exit_{\EE*_{\le i}\FF}\times\exit_{\EE*_{\le j}\FF}$ are mutually conflicting with respect to the conflict relation of $\EE*_{\le \max(i,j)}\FF$. 
\qed
\end{proof}

We end this section by observing that $(\BES,+,\cdot,\|,0,1)$ is a concurrent quantale where the operator $\bullet\in\{\cdot,\|\}$ is redefined so that $\EE\bullet 0 = 0\bullet \EE = 0$. Following Gischer~\cite{Gis88}, we define an order relation based on pomset language subsumption. Recall that a \emph{pomset} is an equivalence class of lposets w.r.t the equivalence relation generated by $\refby_s$. For finite lposets $u$ and $v$, we have $u\refby_s v$ and $v\refby_s u$ iff $\hat u$ is isomorphic to $\hat v$; hence our definition coincides with Gischer's. The equivalence class of a lposet $u$ is denoted by the totally labelled lposet $\hat u$. The pomset language of a BES $\EE$ is defined  by 
$$\{\hat v\ |\ \exists u\in\L(\EE):v\refby_s u\wedge v\textrm{ is a lposet}\}.$$ 
When a BES is considered modulo pomset language equivalence, we show that $(\BES,+,\cdot,0,1)$ and $(\BES,+,\|,0,1)$ are quantales, 
i.e., each structure is an idempotent
semiring, a complete lattice
under the natural order $\EE\leq \EE$ iff $\EE+\FF = \FF$ and 
the operator $\bullet\in\{\cdot,\|\}$ distributes over arbitrary suprema and infinima. The interchange law~(\ref{eq:interchange-law}) is ensured by the subsumption property. The following proposition essentially follows from Gischer's results~\cite{Gis88}. In fact, Gischer proves that the axioms of CKA without the Kleene star completely axiomatise the pomset language equivalence.

\begin{proposition}\label{pro:cka}
For each $\bullet\in\{\cdot,\|\}$, the structure $(\BES,+,\bullet,0,1)$ is a quantale under the pomset language equivalence.
\end{proposition}

\section{Probabilistic Bundle Event Structures}\label{pBES}

In this section, we adapt Katoen's and Varacca's works on probabilistic event structures~\cite{Kat96,Var03}. In particular, we refine the notions of \emph{cluster} and \emph{confusion freeness} which are necessary for the definition of probabilistic bundle event structures (pBES). We use the standard transformation of prime event structures into BES to ensure that our definitions properly generalise Varacca's.

\subsection{Immediate Conflict, Clusters and Confusion Free BES}

The key idea of probabilistic event structures is to use probability as a mechanism to resolve conflicts. However, not all conflicts can be resolved probabilistically~\cite{Kat96}. The cases where this occurs are referred to as confusions. A typical example of confusion is depicted by the first three events $e_1,e_2$ and $e_3$ of Fig~\ref{fig:example} where $e_1\#e_2$, $e_2\#e_3$ and $\neg e_1\#e_3$ hold allowing $e_1$ and $e_3$ to occur simultaneously in a single configuration. However, if the conflict $e_1\# e_2$ is resolved with a coin flip and if the result is $e_2$, then $e_2\# e_3$ cannot be resolved probabilistically because it may produce $e_3$. Following Varacca~\cite{Var03}, we start by characterising conflicts that may be resolved probabilistically.

\begin{definition}\label{def:immediate-conflict} 
Given a BES $\EE$, two events $e,e'\in E$ are in immediate conflict if $e\# e'$ and there exists a configuration $x$ such that $x\cup\{e\}$ and $x\cup\{e'\}$ are again configurations. We write $e\#_\mu e'$ when $e$ and $e'$ are in immediate conflict.
\end{definition}

\begin{example}
In the BES of Fig.~\ref{fig:example}, $e_4$ and $e_5$ are in immediate conflict because $\{e_1,e_3,e_4\}$ and $\{e_1,e_3,e_5\}$ are configurations. In fact, every conflicts in that BES are immediate. Notice that the conflict $e_4\#e_5$ is resolved  when $e_2$ occurs. 
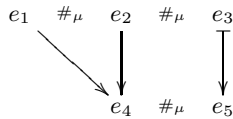
\begin{figure}[h!]
\begin{displaymath}
\xymatrix{
e_1\ar@{->}[dr]\ar@{}[r]|{\#_\mu} & e_2\ar@{->}[d]\ar@{}[r]|{\#_\mu}  & e_3\ar@{|->}[d] \\
& e_4\ar@{}[r]|{\#_\mu}&e_5
}
\end{displaymath}
In this BES, the bundles are $\{e_1,e_2\}\mapsto e_4$ and $\{e_3\}\mapsto e_5$. The conflict relation is $e_1\# e_2$ and $e_2\# e_3$. Therefore, $e_1$ and $e_3$ are concurrent. An arrow $\rightarrow$ represents some part of a bundle (i.e. $\{e_1,e_2\}\mapsto e_4$ is the completed bundle) and $\mapsto$ represents a bundle.
\caption{Immediate conflict in a BES.}\label{fig:example}
\end{figure}
\end{example}

Events can be grouped into clusters of events that are pairwise in immediate conflict. More precisely, we define a cluster as follow.
\begin{definition}\label{def:cluster}
A partial cluster is a set of events $K\subseteq E$ satisfying
\begin{eqnarray}
\forall e,e'\in K: e\neq e\rq{}&\Rightarrow &e\#_\mu e'\qquad\textrm{ and }\nonumber\\
\forall e,e'\in K,x\subseteq E: x\mapsto e&\Rightarrow &x\mapsto e'\nonumber\label{eq:equally-pointed}
\end{eqnarray}
A cluster is a maximal partial cluster (w.r.t inclusion).
\end{definition}

Given an event $e\in E$, the singleton $\{e\}$ is a partial cluster. Therefore, there is always at least one cluster (i.e. maximal) containing $e$ and we write $\<e\>$ the intersection of all clusters containing $e$. 

\begin{example}
In Fig.~\ref{fig:example}, $\{e_1,e_2\}$ and $\{e_2,e_3\}$ are clusters and $\<e_2\> = \{e_2\}$. 
\end{example}

\begin{proposition}
A partial cluster $K$ is maximal (i.e. a cluster) iff
$$
\forall e\in E:(\forall e'\in K:e\#_\mu e'\ \wedge\ \forall x\subseteq E: x\mapsto e\Leftrightarrow x\mapsto e')\Rightarrow e\in K\nonumber
$$
\end{proposition}

\begin{proof}
The forward implication follows from Definition~\ref{def:cluster} and maximality of $K$.
Conversely, assume that $K$ is a partial cluster satisfying the above property. Let $H$ be a partial cluster such that $K\subseteq H$ and $e\in H$. Then, for all $e'\in K$, $e\#_\mu e'$ and 
$$\forall z\subseteq E:x\mapsto e\Leftrightarrow x\mapsto e'$$ 
because $H$ is a partial cluster. By the hypothesis, $e\in K$ and hence $H = K$.\qed
\end{proof}

As in Katoen's and Varacca's works, clusters are used to carry probability and they can be intuitively seen as providing a choice between events where the chosen event happens instantaneously. Notice that our notion of cluster is weaker than Katoen's original definition~\cite{Kat96}: the BES in Fig.~\ref{fig:not-katoen-cluster} contains three clusters $\{e_1,e_2\}$, $\{e_3\}$ and $\{e_4,e_5\}$ and only $\{e_1,e_2\}$ satisfies Katoen's definition.\\
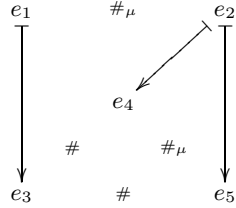
\begin{figure}[h!]
\begin{displaymath}
\xymatrix{
\ar@{|->}[dd]e_1\ar@{}[rr]|{\#_\mu}& & e_2\ar@{|->}[dl] \ar@{|->}[dd]\\
 &e_4 \ar@{}[dl]|\#\ar@{}[dr]|{\#_\mu}&  \\
e_3\ar@{}[rr]|\#&&e_5
}
\end{displaymath}
\caption{A BES where $\{e_1,e_2\}$, $\{e_3\}$ and $\{e_4,e_5\}$ are clusters.}\label{fig:not-katoen-cluster}
\end{figure}

\begin{definition}\label{def:confusion-free}
A BES $\EE$ is confusion free if for all events $e,e'\in E$, 
\begin{itemize} 
\item if $e\#_\mu e'$ then $e\in\<e'\>$, and
\item  if $\<e\>\cap x = \emptyset$ and $x\cup\{e\}\in\C(\EE)$ for some configuration $x\in\C(\EE)$, then $x\cup\{e''\}\in\C(\EE)$ for all events $e''\in\<e\>$.
\end{itemize}
\end{definition}
The first property implies that $\<e\>$ contains all events in immediate conflict with $e$ and hence the confusion introduced by $e_1,e_2$ and $e_3$ in Fig.~\ref{fig:example} is avoided. The second property says that once one event in $\<e\>$ is enabled then all events in $\<e\>$ are also enabled. 
Hence, confusion freeness ensures that all conflicts in $\<e\>$ can be resolved probabilistically regardless of the history. The proof of the following proposition is the same as for prime event structures~\cite{Var03}.

\begin{proposition}\label{pro:confusion-free-theorem}
For a confusion free BES $\EE$, the set $\{\<e\>\ |\ e\in E\}$ defines a partition of $E$. That is, the reflexive closure of $\#_\mu$ is an equivalence relation and the equivalence classes are of the form $\<e\>$.
\end{proposition}

The second property of Definition~\ref{def:confusion-free} is usually hard to check. We give a static and simpler sufficient condition for confusion freeness.
\begin{proposition}\label{pro:confusion-free-simpler}
If a BES $\EE$ satisfies 
$$\forall e,e'\in E:(e\#_\mu e'\Rightarrow e\in\<e'\>)\wedge( \<e\>\cap\cfl(e')\neq\emptyset\Rightarrow \<e\>\subseteq\cfl(e'))$$
then it is confusion free.
\end{proposition}

The second argument of the conjunction says that if some event in $\<e\>$ is in conflict with an event $e'$ then all events in $\<e\>$ are in conflict with $e'$.

\begin{proof}
Let $e\in E$ and $x\in \C(\EE)$ such that $\<e\>\cap x = \emptyset$ and $x\cup\{e\}\in\C(\EE)$. Let $e'\in\<e\>$ and $z\mapsto e'$ be a bundle of $\EE$. We need to show that $x\cup\{e'\}\in\C(\EE)$. By Definition~\ref{def:cluster}, $z\mapsto e$ is also a bundle and since $x$ and $x\cup\{e\}$ are configurations, $e_1\cdots e_n e$ is again a linearisation of $x\cup\{e\}$ for every linearisation $e_1\cdots e_n$  of $x$. Therefore, $z\cap\{e_1,\dots,e_n\}\neq\emptyset$. If $e'\in\cfl(e_i)$ for some $i$, then $\<e\>\subseteq\cfl(e_i)$ by the hypothesis and hence $e\in\cfl(e_i)$, which is impossible because $x\cup\{e\}$ is a configuration. Hence $e_1\cdots e_ne'$ is an event trace, that is, $x\cup\{e'\}\in\C(\EE)$.\qed
\end{proof}

\begin{example}
Fig.~\ref{fig:not-katoen-cluster} depicts a confusion free BES that satisfies Proposition~\ref{pro:confusion-free-simpler}.
\end{example}

With confusion freeness, we are now able to define probability distributions supported by clusters. Recall that a probability distribution on the set $E$ is a function $p:E\to[0,1]$ such that $\sum_{e\in E}p.e = 1$. We say that $p$ is a \textit{probability distribution on $\EE$} if $\supp(p)\subseteq\<e\>$ for some event $e$.

\begin{definition}
A probabilistic BES is a tuple $(\EE,\pi)$ where $\EE$ is a confusion free BES and $\pi$ is a set of probability distribution on $\EE$ such that for every $e\in E$, there exists $p\in\pi$ such that $e\in\supp(p)$.
\end{definition}

The intuition behind this definition is simple: if there is no $p\in\pi$ such that $e\in\supp(p)$ then $e$ is an impossible event and  it can be removed (this may affect any event $e'$ such that $e\preceq_x e'$ for some $x\in\C(\EE)$). Our approach differs from both Varacca's~\cite{Var03} and Katoen's~\cite{Kat96} in that nondeterminism is modelled concretely as a set of probabilistic choices. This approach will mainly contribute to the definition of the probabilistic choice operator $\pc{\alpha}$ of Section~\ref{pcka}. For instance, the expression $a+(b\pc{\alpha} c)$ does not have any meaning in Katoen's pBES, however, it will have a precise semantics in our case.

\section{Probabilistic Simulation on pBES}\label{simulation}

The weakest interpretation of $\refby$ on pBES is the configuration distribution equivalence~\cite{Var03}. However, as in the interleaving case, that is not a congruence~\cite{Seg95}. We use probabilistic simulations which are based on the notion of lifting from~\cite{Den07a}. We denote by $\D(X)$ the set of (discrete) probability distributions over the set $X$. Given $x\in X$, we denote by $\delta_x$ the point distribution concentrated at $x$.

Let $S\subseteq X\times\D(Y)$ be a relation. The lifting of $S$ is a relation $\ov{S}\subseteq \D(X)\times\D(Y)$ such that $(\Delta,\Theta)\in\ov{S}$ iff
\begin{itemize}
\item $\Delta = \sum_{i}\alpha_i\delta_{x_i}$ where $\sum_i\alpha_i = 1$,
\item for every $x_i$, there exists $\Theta_i\in\D(Y)$ such that $(x_i,\Theta_i)\in S$,
\item $\Theta = \sum_{i}\alpha_i\Theta_i$.
\end{itemize}

Notice that the decomposition of $\Delta$ may not be unique. The main properties of lifting are summarised in the following proposition.

\begin{proposition}[\cite{Den07a}]\label{pro:lifting} 
Let $S\subseteq X\times \D(Y)$ be a relation and $\sum_i\alpha_i = 1$. We have
\begin{itemize}
\item if $(\Delta_i,\Theta_i)\in\ov S$ then $(\sum_i\alpha_i\Delta_i,\sum_i\alpha_i\Theta_i)\in\ov S$,
\item if $(\sum_i\alpha_i\Delta_i,\Theta)\in\ov S$ then there exists a collection of distributions  $\Theta_i$ such that $(\Delta_i,\Theta_i)\in\ov S$ and  $\Theta = \sum_i\alpha_i\Theta_i$.
\end{itemize}
\end{proposition}

Since the notion of configuration for a pBES $(\EE,\pi)$ is independent of $\pi$, we keep the notation $\C(\EE)$ for the set of all finite configurations. An example of relation on $\C(\EE)\times\D(\C(\EE))$ is given by the probabilistic prefixing. We say that $x\in\C(\EE)$ is a \textit{prefix} of $\Delta\in\D(\C(\EE))$, denoted (again) by $x\prefix\Delta$, if there exists $p\in\pi$ such that $\supp(p)\cap x = \emptyset$ and $\Delta = \sum_{e\in \supp(p)}(p.e) \delta_{x\cup\{e\}}$. In particular, if $\<e\>  = \{e\}$, $e\notin x$ and $x\cup\{e\}\in\C(\EE)$ then $x\prefix\delta_{x\cup\{e\}}$.  

The relation $\prefix$ is lifted to $\ov{\prefix} \subseteq\D(\C(\EE))\times\D(\C(\EE))$ and the reflexive transitive closure of the lifted relation is denoted by $\ov{\prefix}^*$. Probabilistic prefixing allows us to construct a \emph{configuration-tree} for every pBES. An example is depicted in Fig.~\ref{fig:configuration-tree}.
\begin{figure}[!ht]
\begin{tiny}
\begin{displaymath}
\xymatrix{
&\emptyset\ar[d]&\\
&\{e\}\ar[d]\ar@{.>}[rd]^{0.2}\ar@{.>}[dl]_{0.8}&\\
\{e,e_2\}\ar[d]&\{e,e_1\}\ar@{.>}[dl]_{0.8}\ar@{.>}[dr]^{0.2}&\{e,e_3\}\ar[d]\\
\{e,e_1,e_2\}\ar[d]&&\{e,e_1,e_3\}\ar[d]\\
\{e,e_1,e_2,f\}& &\{e,e_1,e_3,f\} 
}
\end{displaymath}
\end{tiny}
The dotted arrows with common source are parts of a probabilistic prefix relation (e.g. $\{e\}\prefix0.8\delta_{\{e,e_2\}} + 0.2\delta_{\{e,e_3\}}$). The events $e,f$ are the delimiters introduced by $\|$.
\caption{The configurations-tree of the pBES $e_1\|(e_2\pc{0.2}e_3)$ ($\pc{0.2}$ is defined later).}\label{fig:configuration-tree}
\end{figure}

To simplify the presentation, we restrict ourselves to BES satisfying $\exit\cap x = \emptyset$ for every bundle $x\mapsto e$, that is, no event is enabled by a final event. This allows a simpler presentation of the preservation of final events by a simulation. Notice that all BES constructed from the operators defined in this paper satisfy that property (details can be found in the appendix).

\begin{definition}\label{def:simulation}
A (probabilistic) simulation from $(\EE,\pi)$ to $(\FF,\rho)$ is a relation $S\subseteq\C(\EE)\times\D(\C(\FF))$ such that
\begin{itemize}
\item $(\emptyset,\delta_\emptyset)\in S$,
\item if $(x,\Theta)\in S$ then for every $y\in\supp(\Theta)$, $x \refby_s y$,
\item if $(x,\Theta)\in S$ and $x\prefix\Delta'$ then there exists $\Theta'\in\D(\C(\FF))$ such that $\Theta\ov{\prefix}^*\Theta'$ and $(\Delta',\Theta')\in\ov{S}$.
\item if $(x,\Theta)\in S$ and $x\cap\exit_\EE\neq\emptyset$ then for every $y\in\supp(\Theta)$ we have $y\cap\exit_\FF\neq\emptyset$.
\end{itemize}
We write $(\EE,\pi)\refby(\FF,\rho)$ if there is a simulation from $(\EE,\pi)$ to $(\FF,\rho)$.
\end{definition}
Indeed, Definition~\ref{def:simulation} is akin to probabilistic forward simulation on automata. The main difference is the use of the implementation relation $x\refby_s y$ which holds iff there exists a label preserving monotonic bijection from $(\hat y,\preceq_{y},\lambda_{y})$ to $(\hat{x},\preceq_{x},\lambda_{x})$. The implementation relation compares partially ordered configurations rather than totally ordered traces, hence, interferences between incomparable or concurrent events are allowed. Another consequence of this definition is that concurrent events can be linearised while preserving simulation.

\begin{proposition}\label{pro:preorder}
$\refby$ is a preorder.
\end{proposition}

The proof is the same as in~\cite{Den07a}, hence, we provide only a sketch.

\begin{proof}[Sketch]
Reflexivity is clear by considering the relation $\{(x,\delta_x)\ |\ x\in\C(\EE)\}$ which is indeed a simulation. If $R,S$ are probabilistic simulations from $(\EE,\pi)$ to $(\FF,\rho)$ and $(\FF,\rho)$ to $(\G,r)$ respectively then we can show, using Proposition~\ref{pro:lifting} and a similar proof as in the interleaving case, that $R\circ\ov{S}$ is a probabilistic simulation from $(\EE,\pi)$ to $(\G,r)$.\qed 
\end{proof}

A major difference from our previous work~\cite{Rab13} is that the event structure approach provides a truly concurrent interpretation of pCKA. The most notable benefit of using a true-concurrent model is substitution~\cite{Gor97,Gis88} where a single step event can be refined with another event structure after a concurrency operator has been applied. In the automata model, such a substitution must occur before the application of the concurrency operator to obtain the correct behaviour. Moreover, in interleaving, concurrency is related to the nondeterministic choice whereas here the two operators are orthogonal. 
\begin{example}
In Fig.~\ref{fig:concurrency}, it is shown that $a\cdot b + b\cdot a\refby a\|b$ but the converse does not hold. 
\end{example}
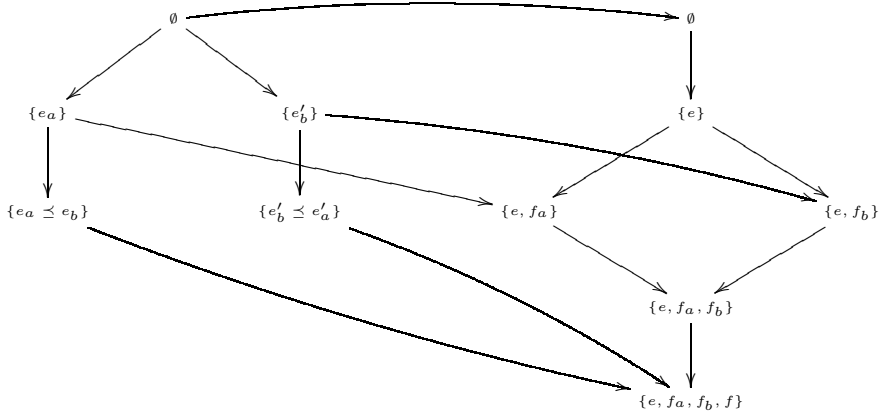
\begin{figure}[!ht]
\begin{tiny}
$$
\hspace{-1cm}\xymatrix{
&\emptyset\ar[dl]\ar[dr]\ar@/^/[rrrr]& &&  		&\emptyset\ar[d]& \\
\{e_a\}\ar[d]\ar[rrrrd]&&\{e_b'\}\ar[d]\ar@/^/[rrrrd] && 				&\{e\}\ar[dl]\ar[dr]&\\
\{e_a\preceq e_b\}\ar@/_/[rrrrrdd]&&\{e'_b\preceq e_a'\}\ar@/^/[rrrdd]				&& \{e,f_a\}\ar[dr]&& \{e,f_b\}\ar[dl]\\
&& && &\{e,f_a,f_b\}\ar[d]& \\
&& && &\{e,f_a,f_b,f\}&
}
$$
\end{tiny}
Since $\{e,f_a,f_b,f\}\not\refby_s\{e_a\preceq e_b\}$ nor $\{e,f_a,f_b,f\}\not\refby_s\{e_b'\preceq e_a'\}$, it is impossible to find a simulation from $a\|b$ to $a\cdot b+b\cdot a$. In the configuration tree on the left, the order $\preceq$ is made explicit and primes are introduced for disjointness.
\caption{A simulation from $a\cdot b+b\cdot a$ to $a\|b$. }\label{fig:concurrency}
\end{figure}

\section{Probabilistic Concurrent Kleene Algebra}\label{pcka}

In this section, we show that the set  $\pBES$ endowed with a nondeterministic choice $(+)$, a probabilistic choice $(\!\pc{\alpha}\!)$, a sequential composition $(\cdot)$, a concurrent composition $(\|)$ and the binary Kleene star $(*)$ satisfy the axioms of Fig.~\ref{fig:axioms}. These axioms are a combination of the basic algebraic laws of CKA~\cite{Hoa09} and pKA~\cite{Mci05}. 
\begin{figure}
\hspace{-1cm}\begin{minipage}{.5\linewidth}
\begin{eqnarray}
\EE + \EE & \equiv & \EE\label{eq:+-idem}\\
\EE + \FF & \equiv & \FF + \EE \\
\EE + (\FF + \G) & \equiv & (\EE + \FF) + \G\label{eq:+-assoc}\\
\EE + 0 & \equiv & \EE \label{eq:+-zero}
\end{eqnarray}
\end{minipage}
\begin{minipage}{.57\linewidth}
\begin{eqnarray}
\EE & \equiv &\EE\pc{\alpha}\EE\label{eq:pc-idem}\\
\EE\pc{\alpha}\FF &\equiv & \FF\pc{1-\alpha}\EE\\
\EE\pc{\alpha}(\FF\pc{\beta}\G) &\equiv & (\EE\pc{\frac{\alpha(1-\beta)}{1-\alpha\beta}}\FF)\pc{\alpha\beta}\G\label{eq:pc-assoc}\\
(\EE\pc{\alpha}\FF)\cdot\G & \equiv & \EE\cdot\G\pc{\alpha}\FF\cdot\G
\end{eqnarray}
\end{minipage}

\hspace{-1cm}\begin{minipage}{.5\linewidth}
\begin{eqnarray}
\EE \cdot (\FF \cdot \G) & \equiv & (\EE \cdot \FF) \cdot \G\label{eq:cdot-assoc}\\
\EE\cdot 1 & \equiv & \EE\label{eq:rone}\\
1\cdot \EE & \equiv & \EE \label{eq:lone}\\
0\cdot\EE& \equiv & 0 \label{eq:seq-zero}
\end{eqnarray}
\end{minipage}
\begin{minipage}{.57\linewidth}
\begin{eqnarray}
1\|\EE &\equiv & \EE\label{eq:par-unit}\\
\EE\| \FF & \equiv & \FF\|\EE\label{eq:par-comm}\\
\EE \| (\FF \| \G) & \equiv & (\EE \| \FF) \| \G\label{eq:par-assoc}
\end{eqnarray}
\end{minipage}

\hspace{-1cm}\begin{minipage}{.5\linewidth}
\begin{eqnarray}
(\EE+\FF)\cdot \G & \equiv & \EE\cdot \G + \FF\cdot \G\label{eq:+-dist-seq} \\
\EE\cdot \FF+\EE\cdot \G & \refby & \EE \cdot (\FF+\G)\label{eq:+-subdist-seq}\\
\EE\cdot(\FF\pc{\alpha}\G)&\refby &\EE\cdot\FF\pc{\alpha}\EE\cdot\G\label{eq:pc-supdist-seq}
\end{eqnarray}
\end{minipage}
\begin{minipage}{.57\linewidth}
\begin{eqnarray}
\EE\| \FF+\EE\| \G & \refby & \EE \| (\FF+\G)\label{eq:+-subdist-par}\\
\EE\|(\FF\pc{\alpha}\G)& \refby&\EE\|\FF\pc{\alpha}\EE\|\G\label{eq:pc-supdist-par}\\
(\EE\|\FF)  \cdot (\EE'\| \FF') & \refby & (\EE\cdot \EE')\| (\FF\cdot \FF')\label{eq:interchange-law}
\end{eqnarray}
\end{minipage}

\begin{eqnarray}
\FF + \EE\cdot(\EE*\FF) & \equiv & (\EE*\FF) \label{eq:unfold}\\
\G + \EE\cdot\FF\refby\FF&\Rightarrow&\EE*\G\refby\FF\label{eq:induction}
\end{eqnarray}

\caption{Axioms of pCKA satisfied by $\pBES$ modulo probabilistic simulation. Here, we write a pBES simply with $\EE$ instead of the tuple $(\EE,\pi)$ and $\alpha\beta<1$ in Equation~(\ref{eq:pc-assoc}) (the case $\alpha\beta=1$ being a simplification of the left hand side).}\label{fig:axioms}
\end{figure}

We generate the pBES $(0,\emptyset),(1,\{\delta_{e}\})$ and $(a,\{\delta_{e_a}\})$ from the basic BES. To simplify the notations, these basic  pBES are again denoted by $0,1$ and $a$. The other operators are defined as follows:
\begin{eqnarray}
(\EE,\pi) + (\FF,\rho) & = & (\EE+\FF,\pi\cup \rho)\nonumber\label{op:nondet}\\
(\EE,\pi)\cdot(\FF,\rho) & = & (\EE\cdot\FF,\pi\cup \rho)\nonumber\label{op:seq}\\
(\EE,\pi)\|(\FF,\rho) & = & (\EE\|\FF,\pi\cup \rho\cup\{\delta_e,\delta_f\})\nonumber\label{op:par}
\end{eqnarray}
where $e$ and $f$ are the fresh events delimiting $\EE\|\FF$. Recall that $\EE$ and $\FF$ are assumed to be disjoint in these definitions. The probabilistic choice that chooses $\EE$ with probability $1-\alpha$ and $\FF$ with probability $\alpha$ is
\begin{eqnarray}
(\EE,\pi)\pc{\alpha}(\FF,\rho) &  = &  (\EE+\FF,\pi\pc{\alpha}\rho)\label{op:pchoice}\nonumber
\end{eqnarray}
where $r\in \pi\pc{\alpha}\rho$ iff:
\begin{itemize}
\item if $\supp(r)\subseteq\init(\EE)\cup\init(\FF)$ then $r = (1-\alpha) p+\alpha q$ for some $p\in\pi$ and $q\in\rho$,
\item else $r\in\pi\cup\rho$.
\end{itemize}

Intuitively, nondeterminism is resolved first by choosing a probability distribution, then a probabilistic choice is resolved based on that distribution. Indeed, the nondeterministic and probabilisic choices introduce clusters.  
\begin{example}
The BES $a\|(b\pc{0.2}c)$ contains four clusters $\<e\>,\<e_b,e_c\>,\<e_a\>$ and $\<f\>$ where $e,f$ are the delimiter events. It has a set of probability distributions $\{0.8\delta_{e_b} + 0.2\delta_{e_c},\delta_{e_a},\delta_e,\delta_f\}$. In contrast, the event structure $a + (b\pc{0.2}c)$ has a single cluster $\<e_a,e_b,e_c\>$ with set of probability distributions $\{0.8\delta_{e_b} + 0.2\delta_{e_c},\delta_{e_a}\}$.
\end{example}

To construct the binary Kleene star, we need the following partial order 
$$(\EE,\pi)\prefix(\FF,\rho)\qquad\textrm{iff}\qquad \EE\prefix\FF\wedge \pi = \{p\in\rho\ |\ \supp(p)\subseteq E\}.$$
The proof that $\prefix$ is indeed $\omega$-complete is essentially the same as in the standard case (Section~\ref{operations}). Hence the Kleene product $(\EE,\pi)*(\FF,\rho)$ is again the limit of the increasing sequence of pBES:
$$(\FF,\rho)\prefix (\FF,\rho) + (\EE,\pi)\cdot(\FF,\rho)\prefix (\FF,\rho)+ (\EE,\pi)\cdot((\FF,\rho) + (\EE,\pi))\prefix\cdots.$$
More precisely, $(\EE,\pi)*(\FF,\rho) = (\EE*\FF,\pi*\rho)$ where $\pi*\rho = \cup_i\pi*_{\le i}\rho$ and each set $\pi*_{\le i}\rho$ is obtained from the construction of $\EE*_{\le i}\FF$.

A BES is \textit{regular} if it is inductively defined with the operators of Section~\ref{operations}. 

\begin{proposition}\label{pro:confusion-free-regular}
A Regular BES is confusion free.
\end{proposition}

\begin{proof}[Sketch]
By induction on the structure of the BES.
\end{proof}

\begin{proposition}\label{pro:precongruence}
The order $\refby$ is a precongruence i.e. for every pBES $(\EE,\pi), (\FF,\rho)$ and $(\G,\eta)$, if $(\EE,\pi)\refby (\FF,\rho)$ then $(\EE,\pi)\bullet(\G,\eta)\refby (\FF,\rho)\bullet(\G,\eta)$ (and symmetrically) for every $\bullet\in\{+,\cdot,\|,*\}$. 
\end{proposition}

\begin{proof}[Sketch]
Let $(\EE,\pi)\refby(\FF,\rho)$ be witnessed by a simulation $S\subseteq\C(\EE)\times\D(\C(\FF))$ and $(\G,\eta)$ be any pBES. The congruence properties are proven by extending the simulation $S$ to the events of $\G$. For instance, That $(\EE,\pi) + (\G,\eta)\refby (\FF,\rho) + (\G,\eta)$ is deduced by showing that $S\cup\{(x,\delta_x)\ |\ x\in\C(\G)\}$ is indeed a simulation. 

\qed
\end{proof}

The axioms~(\ref{eq:+-idem}-\ref{eq:seq-zero}) and~(\ref{eq:par-comm}-\ref{eq:+-dist-seq}) are proven using simulations akin to the interleaving case~\cite{Rab13,Den07a}. The existence of simulations that establishes axiom~(\ref{eq:par-unit})  is clear from the definition of $\|$ and $1$. It follows from the axioms of $+$ and Proposition~\ref{pro:precongruence} that $(\EE,\pi)\refby (\FF,\rho)$ if and only if $(\EE,\pi) + (\FF,\rho)\equiv(\FF,\rho)$.

\begin{proposition}\label{pro:subdistributivity}
The axioms~(\ref{eq:+-subdist-seq},\ref{eq:pc-supdist-seq}) and (\ref{eq:+-subdist-par},\ref{eq:pc-supdist-par}) and the interchange law~(\ref{eq:interchange-law}) hold on $\pBES$ modulo probabilistic simulation.
\end{proposition}

\begin{proof}[Sketch]
These equations are proven by the usual simulation constructions.
\qed
\end{proof}

\begin{proposition}\label{pro:kleene-star}
The binary Kleene star satisfies the axioms~(\ref{eq:unfold}) and~(\ref{eq:induction}).
\end{proposition}
\begin{proof}[Sketch]
The first equation is proven using the standard simulation construction. For the second one, let $S\subseteq\C(\EE\cdot\FF)\times\D(\C(\FF))$ be a probabilistic simulation from $(\EE,\pi)\cdot(\FF,\rho)$ to $(\FF,\pi)$. By monotonicity of $\cdot$ and $+$, there exists a simulation $S^{(i)}\subseteq\C(\EE*_{\le i}\FF)\times\D(\C(\FF))$ from $(\EE,\pi)*_{\le i}(\FF,\rho)$ to $(\FF,\rho)$, for every $i\in\N$. Moreover, we can find a family of simulations such that $S^{(i-1)}$ is the restriction of $S^{(i)}$ to $(\EE,\pi)*_{\le i-1}(\FF,\rho)$. Thus, we can consider the reunion $S = \cup_i S^{(i)}$ and show that it is indeed a simulation from $(\EE,\pi)*(\FF,\rho)$ to $(\FF,\rho)$. Hence, Equation~(\ref{eq:induction}) holds.\qed
\end{proof}

\begin{theorem}
The set $\pBES$ modulo probabilistic simulation forms a probabilistic concurrent Kleene algebra with a binary Kleene star.
\end{theorem}

\section{Conclusion}

We have constructed a truly concurrent model for probabilistic concurrent Kleene algebra using pBES. In the process, we also set out a notion of probabilistic simulation for these event structures. The semantics of pBES was defined by constructing the configuration-trees using prefixing and probabilistic simulations are exhibited when possible. Since the simulation distinguishes between concurrency and interleaving, we believe that it provides a suitable combination of nondeterminism, probability and true-concurrency.

Our main result is the soundness of pCKA axioms. The completeness of such an axiom system is still open. We believe that other axioms such as guarded tail recursion are needed to achieve a complete characterisation as in~\cite{Seg04}. Another interesting specialisation of this work is the labelling of events with one-step probabilistic programs. These however require further studies.

\bibliographystyle{splncs}
\bibliography{lpar19-paper69}

\newpage
\section*{Appendix}

\renewenvironment{theorem}[2][Theorem]{\begin{trivlist}
\item[\hskip \labelsep {\bfseries #1}\hskip \labelsep {\bfseries #2}]}{\end{trivlist}}
\newenvironment{repproposition}[2][Proposition]{\begin{trivlist}
\item[\hskip \labelsep {\bfseries #1}\hskip \labelsep {\bfseries #2}]}{\end{trivlist}}

In this appendix, we denote event traces simply by the Greek letters $\alpha,\beta,\dots$ and $\ov\alpha$ is the set of events occurring in the event trace $\alpha$. 

\subsection{Proof Complement for Proposition~\ref{pro:configuration-prefix}}

\begin{lemma}\label{pro:trace-restriction}
Let $\alpha\in\TT(\EE)$ and $x\in\C(\EE)$ such that $x\subseteq\bar\alpha$, then the restriction $\alpha|_{x}$ of $\alpha$ to events in $x$ is an event trace.
\end{lemma}

\begin{proof}
Let $\alpha = e_1e_2\cdots e_n$, $x\in\C(\EE)$ and $\alpha|_{x} = e_{i_1}e_{i_2}\cdots e_{i_m}$. Let $e_{i_k}\in x$ and $z\mapsto e_{i_k}$ be a bundle of $\EE$. Since $\alpha$ is an event trace, there exists a even $e_j$ such that $e_j\in z$ and $j<i_k$. Since $x$ is a configuration and $e_{i_k}\in x$, there exists $e_{i_l}\in z$ and $l<k$. By definition, the bundle set $z$ contains mutually conflicting events only and since $y$ is conflict free, $e_{i_l} = e_j$. Hence, $\alpha|_{x}$ is an event trace.  \qed
\end{proof}

\begin{lemma}\label{pro:trace-extension}
Let $x\in\C(\EE)$ and $y\in\C(\EE)$ such that $x\subseteq y$, for every trace event $\alpha$ such that $\ov\alpha = x$ there exists a trace event $\alpha'$ such that $\ov{\alpha'} = y$ and $\alpha'|_{x} = \alpha$.
\end{lemma}

\begin{proof}
Let $\alpha,\beta$ be any event traces such that $\ov\alpha = x$, $\ov{\beta} = y$ and $x\subseteq y$. Let $\beta'$ bet the concatenation of two sequences $\beta_1\beta_2$ where events in $\beta_1$ are exactly those of $x$ ordered with $\preceq_\beta$ and $\beta_2$ is composed of events from $y\setminus x$ ordered again with $\preceq_\beta$. We now show that $\beta'$ is an event trace. That $\beta'$ is conflict free comes from the configuration $y$. Let $z\mapsto e_1$ be a bundle of $\EE$ such that $e_1\in \beta_1$. Since $\ov\beta_1 = x$ is a configuration, $z\cap \beta_1\neq\emptyset$ and that element has to be ordered before $e_1$ with respect to $\preceq_\beta$ because $z$ contains mutually conflicting events so $z\cap\beta_1 = z\cap\beta$ contains exactly one event. That is, $\beta_1$ is an event trace. As for $\beta_2$, let $z\mapsto e_2$ be a bundle and $e_2\in\beta_2$. Since $y$ is a configuration, we have $z\cap\beta\neq\emptyset$ and the sole event in that intersection is ordered before $e_2$ in the event trace $\beta_1\beta_2$ because $\preceq_{\beta_2}\subseteq\preceq_\beta$. Hence $\beta'$ is an event trace. 

Finally, let $\beta'' = \alpha\beta_2$. With the same argument as before, we can show that $\beta''$ is an event trace and hence $\beta''|_{x} = \alpha$.\qed
\end{proof}

\begin{repproposition}{\ref{pro:configuration-prefix}.}
Let $\EE$ be a BES, if $x\in\C(\EE)$, $y\in\C(\EE)$ and $x\subseteq y$ then $(x,\preceq_x,\lambda_x)\prefix (y,\preceq_y,\lambda_y)$.
\end{repproposition}

\begin{proof}
Let $x\subseteq y$. Let us first show that $\preceq_x = \preceq_y\cap(x\times x)$. Let $e,e'\in x$ such that  $e\preceq_xe'$. Lemma~\ref{pro:trace-restriction} implies that $e\preceq_y e'$ because every event trace for $y$ restricts to an event trace for $x$. For the converse inclusion, let $e,e'\in x$ such that $e\preceq_ye'$. Lemma~\ref{pro:trace-extension} implies that every event trace for $x$ can be obtained as a restriction of some event trace for $y$. Hence, $e\preceq_x e'$. Therefore $\preceq_x = \preceq_y\cap(x\times x)$.

Let $e,e'\in y$, $e\preceq_y e'$ and $e'\in x$. It now suffices to show that $e\in x$. In fact, if $e\notin x$, then there exists an event trace $\beta' = \beta_1\beta_2$ as specified in the proof of Lemma~\ref{pro:trace-extension}, that is, $\ov{\beta'} = y$, $\ov{\beta_1} = x$ and $e\in\ov{\beta_2}$. Therefore, $e\npreceq_{\beta'} e'$ which contradict the fact that $e\preceq_ye'$. \qed
\end{proof}

\subsection{Proof Complement for Proposition~\ref{pro:omega-completeness} and Properties of $\prefix$}

\begin{repproposition}{\ref{pro:omega-completeness}}
$(\BES,\prefix)$ is an $\omega$-complete partially ordered set.
\end{repproposition}

\begin{proof}
Firstly, we prove that $\prefix$ is a partial order. It is clear that $\prefix$ is reflexive. To prove antisymmetry, Let $\EE\prefix\EE'$ and $\EE'\prefix\EE$, then $E = E'$. Let $z\mapsto e$ is a bundle of $\EE$. Since $E = E'$, we have $e\in E'$ and hence $z\subseteq E'$ and $z\mapsto' e$ i.e. it is also a bundle of $\EE'$. The fact that $\#$ and $\#'$ (resp. $\lambda$ and $\lambda'$) coincide follows directly from the definition. To prove transitivity, let $\EE\prefix\EE'$ and $\EE'\prefix\EE''$. We need to prove that $\EE\prefix\EE''$. It is clear that $E\subseteq E''$. Let  $z\mapsto'' e$ be a bundle of $\EE''$ and $e\in E$. Since $E\subseteq E'$, we have $e\in E'$ and since $\EE'\prefix\EE''$, we obtain $z\subseteq E'$ and $z\mapsto' e$ is a bundle of $\EE'$. Since $e\in E$ and $\EE\prefix \EE'$, we have $z\subseteq E$ and $z\mapsto e$ is a bundle of $\EE$. The properties $\# = \#''\cap E\times E$ and $\lambda = \lambda''\cap E\times\Sigma$ and $\exit = \exit''\cap E$ follows from similar argument. Hence $\EE\prefix\EE''$.

Secondly, let $\EE_0\prefix\EE_1\prefix\EE_2\prefix\cdots$ be a countable increasing chain of BES and let $\EE = \cup_i\EE_i$ endowed with the following components:
\begin{itemize}
\item set of events: $E = \cup_i E_i$,
\item conflict relation: $\# = \cup_i\#_i$,
\item bundle relation: $\mapsto = \cup_i\mapsto_i$,
\item labelling function: $\lambda = \cup_i\lambda_i$,
\item final events: $\exit = \cup_i\exit_i$.
\end{itemize}
We show that $\EE_i\prefix\EE$ for all $i$ and if $\EE_i\prefix\FF$ for all $i$ then $\EE\prefix\FF$.

Let $i\in\N$, we have $E_i\subseteq E$ by construction. Let $(e,e')\in\#\cap (E_i\times E_i)$. Since $\# = \cup_i\#_i$, there exists $j\in\N$ such that $e\#_je'$. There are two cases:
\begin{itemize}
\item if $j\leq i$, then $\#_j\subseteq\#_i$ and $e\#_ie'$,
\item if $i<j$, then $\EE_i\prefix \EE_j$ and hence $\#_j\cap(E_i\times E_i) = \#_i$ i.e. $e\#_i e'$.
\end{itemize}
A similar argument can be used to prove $\lambda_i = \lambda\cap(E_i\times\Sigma)$, $\exit_i = \exit \cap{E_i}$ and the relationship between bundles of $\EE$ and $\EE_i$.

Finally, let $\EE_i\prefix \FF$ for all $i$. We need to show that $\EE\prefix\FF$. It is clear that $E\subseteq F$ where $F$ is the set of events of $\FF$. 
\begin{itemize}
\item Let $(e,e')\in\#_\FF\cap(E\times E)$. By definition of $E$, there exists $i,j\in\N$ such that $e\in E_i$ and $e'\in E_j$. Assume that $i\leq j$, then $e,e'\in E_j$. Since $\EE_j\prefix\FF$, we have $e\#_je'$ and hence $e\#e'$. 
\item A similar argument can be used to prove $\lambda = \lambda_\FF\cap(\EE\times\Sigma)$ and $\exit = \exit_F\cap E$.
\item It is clear that $\mapsto\subseteq\mapsto_\FF$. Let $z\mapsto_\FF e$ be a bundle of $\FF$ and $e\in E$. There exists $i\in\N$ such that $e\in E_i$ and since $\EE_i\prefix\FF$, we deduce that $z\subseteq E_i$ and $z\mapsto_i e$ is a bundle of $\EE_i$. Hence, $z\mapsto e$ is a bundle of $\EE$.\qed
\end{itemize}
\end{proof}

\begin{proposition}\label{pro:prefix-trace}
Let $\EE,\EE'$ be two BES such that $\EE\prefix\EE'$, then $\TT(\EE) = \{\alpha\ |\ \alpha\in\TT(\EE')\wedge \ov\alpha\subseteq E\}$.
\end{proposition}

\begin{proof}
Let $\alpha\in\TT(\EE')$ such that $\ov\alpha\subseteq E$ and let us show that $\alpha\in\TT(\EE)$. Let us write $\alpha = e_1e_2\cdots e_n$. By definition of an event trace, we have $z\cap\{e_1,\dots,e_{i-1}\} \neq\emptyset$ for every bundle $z\mapsto'e_i$ in $\EE'$ and since $\mapsto\subseteq\mapsto'$ we also have $y\cap\{e_1,\dots,e_{i-1}\}\neq\emptyset$ for every bundle $y\mapsto e_i$ in $\EE$. On the other hand, since $\#'\cap E\times E = \#$ and $\alpha$ is an event trace of $\EE'$, we have $e\notin\cfl_\EE(\{e_1,\dots,e_{i-1}\})\cup\{e_1,\dots,e_{i-1}\}$ and therefore, $\alpha\in\TT(\EE)$.

Conversely, let $\alpha = e_1e_2\cdots e_n\in\TT(\EE)$, we need to show that $\alpha\in\TT(\EE')$. Let $z\mapsto'e_i$ be a bundle of $\EE'$ where $e_i\in\ov\alpha$. Since $\EE\prefix\EE'$ and $e_i\in E$, we have $z\mapsto e$ is a bundle of $\EE$ and therefore $z\cap\{e_1,\dots,e_{i-1}\}\neq\emptyset$. Moreover, since $\#'\cap E\times E = \#$ and $\ov\alpha\subseteq E$, we deduce that $e_i\notin\cfl_{\EE'}(\{e_1,\dots,e_{i-1}\}$. Lastly, that $e_i\notin\{e_1,\dots,e_{i-1}\}$ follows directly from the fact that $\alpha$ is an event trace. Hence, $\alpha\in\TT(\EE')$.\qed
\end{proof}

\begin{corollary}\label{cor:prefix-refinement}
if $\EE\prefix\EE'$ then $\C(\EE) \subseteq \C(\EE')$ and $\L(\EE)\subseteq\L(\EE')$.
\end{corollary}

\begin{corollary}\label{cor:limit-lposet}
If $\EE_0\prefix\EE_1\prefix\EE_2\prefix\cdots$ is a increasing family of BES with limit $\cup_i\EE_i$ then $\L(\cup_i\EE_i) = \cup_{i\in\N}\L(\EE_i)$.
\end{corollary}

\begin{proof}
It is clear from Corollary~\ref{cor:prefix-refinement} that $\L(\EE_i)\subseteq\L(\cup_i\EE_i)$ and hence $\cup_i\L(\EE_i)\subseteq\L(\cup_i\EE_i)$.

Conversely, let $u\in\L(\cup_i\EE_i)$. By definition, $u = (x,\preceq_x,\lambda_x)$ where $x\in\C(\cup_i\EE_i)$. By construction, the set of events of $\cup_i\EE_i$ is $ \cup_i E_i$ where each $E_i$ is the set of events of the BES $\EE_i$. Since $x$ is a finite subset of $\cup_iE_i$ and $E_0\subseteq E_1\subseteq E_2\subseteq\cdots$, there exists $j\in\N$ such that $x\subseteq\ E_j$ and hence $u\in\L(\EE_j)$ follows from Proposition~\ref{pro:prefix-trace} (that is, the order $\preceq_x$ obtained from the BES $\cup_i\EE_i$ coincides with the order obtained from the BES $\EE_j$). \qed
\end{proof}

\subsection{Correspondence between our work, Varacca's and Katoen's}

In this subsection, given a PES $(E,\#,\leq,\lambda)$ we denote by $[e] = \{e' \ |\ e'\leq e\}$ and $[e) = [e]\setminus\{e\}$.

\begin{proposition}\label{pro:immediate-conflict}
Given an PES $\EE$ and its corresponding BES $\EE'$, then for every $e,e'\in E$, $e\#_\mu e'$ in $\EE$ iff $e\#_\mu e'$ in $\EE'$.
\end{proposition}

\begin{proof}
Let $e\#_\mu e'\in E$, then $[e]\cup[e')$ and $[e']\cup[e)$ are configurations. Since $e$ and $e'$ are respectively maximal in these two configurations, we have $x = [e)\cup[e')\in\C(\EE')$ which satisfies Definition~\ref{def:immediate-conflict}.

Conversely, assume that $x,x\cup\{e\},x\cup\{e'\}\in\C(\EE')$ and $e\#e'$. By definition of $E'$, $[e)\subseteq x$ and $[e')\subseteq x$. Therefore, $[e)\cup [e')$ is down-closed and does not contain any conflicting elements. Since $x\cup\{e\}\in\C(\EE)$, $e$ is not in conflict with any element of $[e')$ and hence $[e)\cup[e')\cup\{e\}\in\C(\EE)$. Similarly, we prove that $[e)\cup[e')\cup\{e'\}\in\C(\EE)$. \qed
\end{proof}

\begin{proposition}
A cluster in the sense of Katoen~\cite{Kat96} satisfies Definition~\ref{def:cluster}.
\end{proposition}

\begin{proof}
Let $K$ be a cluster in the sense of Katoen's~\cite{Kat96} and $e,e'\in K$. Let $x$ be a configuration such that $x\cup\{e\}$ is a configuration. Then for every bundle $z\mapsto e'$, we have $x\cap z\neq \emptyset$ because $z\mapsto e'$ iff $z\mapsto e$. Moreover, $x\cap\cfl(e')=\emptyset$ else that elements should be in the cluster $K$ and hence conflicting with event $e$ too (this is impossible because $x\cup\{e\}$ is a configuration). Therefore, $x\cup\{e'\}$ is a configuration and $e\#_\mu e'$. 

The second property of partial clusters (events of a clusters are equally pointed) is found in Katoen's definition and it suffices to prove that $K$ is maximal. Let $H$ be a partial cluster such that $K\subseteq H$ and $e\in H$. By definition of a partial cluster, $e\#_\mu e'$ for every $e'\in K$. Since $e\#_\mu e'$ implies $e\#e'$, we deduce that $e\in K$ because Katoen's cluster contains every event that is conflicting with all events in it.\qed
\end{proof}

\begin{proposition}\label{pro:cluster-cell}
Definition~\ref{def:cluster} coincides with Varacca's partial cells on PES~\cite{Var03}. In particular, a cluster corresponds to a cell on PES.
\end{proposition}

\begin{proof}
Let $\EE = (E,\leq,\#,\lambda)$ be a PES and $\EE' = (E,\mapsto,\#,\lambda)$ its corresponding BES. 

Let $K$ be a cluster of $\EE'$ as per Definition~\ref{def:cluster} and $e,e'\in K$. We show that $K$ is a cell. It follows directly from the definition of transformation that that $[e)=[e')$ and $e\#_\mu$. 

Conversely, let $K$ be a cell of $\EE$ and $e,e'\in K$. It contains mutually immediate conflicting events by definition of a cell. Since, $[e) = [e)'$, we deduce that $z\mapsto e$ implies $z\mapsto e'$ for every $z\subseteq E$. \qed
\end{proof}

Remind that a PES is confusion free if and only if $\#_\mu$ is transitive and $e\#_{\mu}e'$ implies $[e)=[e')$.

\begin{proposition}
A PES is confusion free iff its corresponding BES is confusion free as per Definition~\ref{def:confusion-free}.
\end{proposition}

\begin{proof}
Let $\EE$ be a PES and $\EE'$ its corresponding BES. 

Assume that $\EE$ is a confusion free PES and $e,e'\in E$ such that $e\#_\mu e'$. Let $K$ be a cluster such that $e'\in K$. By definition of a confusion free PES, $e\#_\mu e'$ implies $[e) = [e')$ i.e. $e$ and $e'$ are pointed by the same bundles. By transitivity of $\#_\mu$, $e\#_\mu e''$ for every $e''\in K$ and therefore $e\in K$ by maximality of $K$. Since that is true for every cluster containing $e'$, we have $e\in\<e'\>$.

We now prove the property of Proposition~\ref{pro:confusion-free-simpler}. Let $e,e'\in E$ such that $e''\in\<e\>\cap\cfl(e')\neq\emptyset$. Since $e''\#e'$, there exists $e_0',e_0''\in E$ such that $e_0''\leq e''$ and $e_0'\leq e'$ and $e_0''\#_mu e_0'$. If $e_0''<e''$ then $e_0''<e$ because the BES is confusion free (and hence $e''$ and $e$ are pointed by the same bundle). Hence $e\#e'$ by heredity of $\#$. Else if $e_0'' = e''$, then $e\#_\mu e_0'$ because $\#_\mu$ is transitive. It follows that $e\#e'$ by heredity of $\#$. 

Conversely, assume that $\EE'$ is a confusion free BES and $e,e',e''\in E$ such that $e\#_\mu e'$ and $e'\#_\mu e''$. The first property of confusion freeness implies that $e\in\<e'\>$ and $e'\in\<e''\>$, that is, $[e)=[e')$ and $[e') = [e'')$. Therefore, $e\#_\mu e''$ holds in the PES $\EE$ and whenever $e_1\#_\mu e_2$ holds in the BES $\EE'$, we have $[e_1) = [e_2)$ i.e. $\EE$ is confusion free.\qed
\end{proof}

\subsection{Complementary Proofs of the Algebraic Laws}

\begin{proposition}\label{pro:exit-maximal}
Let $\EE$ be a regular BES. If a configuration $x\in\C(\EE)$ contains a final event i.e. $x\cap\exit_\EE \neq\emptyset$ then it is maximal. 
\end{proposition}

\begin{proof}
We reason by structural induction. The claim holds for the basic BES. Let $\EE,\FF$ be two BES satisfying the induction hypothesis. 

The case of $+$ is clear because $\C(\EE+\FF) =\C(\EE)\cup\C(\FF)$. 

For $(\cdot)$, let $x\in\C(\EE\cdot\FF)$ and assume that $x$ contains an exit event. Since all events in $\exit_\EE$ occurs before any event of $\FF$, we have $x\cap E\in\C(\EE)$ (Proposition~\ref{pro:prefix-trace}). Moreover, we show similarly that $x\cap F\in\C(\FF)$  and is maximal by induction hypothesis. Since $x\cap F$ must at least contain one event from $\init(\FF)$, $x\cap E$ necessarily contain an event from $\exit_\EE$ and hence $x\cap E$ is also maximal in $\EE$. Therefore, $x$ must be maximal in $\EE\cdot\FF$. 

For $\|$, let $x\in\C(\EE\|\FF)$ and $e,f$  be the fresh events introduced by the construction of $\EE\|\FF$. Here again, we have $x\cap E\in\C(\EE)$ and $x\cap F\in\C(\FF)$ and since $\exit_{\EE\|\FF} = \{f\}$ and $\exit_\EE\mapsto f$ and $\exit_\FF\mapsto f$, by induction hypothesis, $x\cap E$ and $x\cap F$ are maximal in $\EE$ and $\FF$ respectively and we deduce the maximality of $x$.

For $*$, if $x\in\C(\EE*\FF)$ then $x\in\C(\EE*_{\le i}\FF)$ for some $i$ and we are back to the case of $(+)$ and $(\cdot)$. \qed
\end{proof}

\begin{corollary}\label{pro:seq-maximal}
Let $\EE,\FF$ be two regular BES, if $x\in\C(\EE\cdot\FF)$ and $x\cap F\neq\emptyset$ then $x\cap E$ is maximal in $\EE$.
\end{corollary}

\begin{repproposition}{\ref{pro:confusion-free-regular}}
Every regular BES is confusion free.
\end{repproposition}

\begin{proof}
We reason by structural induction on the structure of $\EE$. It is clear that the basic BES are confusion free and the first property of confusion freeness follows directly from the fact that two events of $\EE_1\bullet \EE_2$, for  $\bullet\in\{+,\cdot,\|\}$, are in immediate conflict if an only if they are in immediate conflict in $\EE_1$ or $\EE_2$; or both belongs to $\init(\EE_1\bullet\EE_2)$ and we have $\init(\EE)\#\init(\EE)$ for every regular BES $\EE$(resp. they are in immediate conflict in $\EE*_{\le i}\FF$ for some $i\in\N$). Let us concentrate on the second property. Let $x\in\C(\EE)$ and $e\in E\setminus x$ such that $x\cup\{e\}\in\C(\EE)$. Let $e'\in\<e\>$.
\begin{itemize}
\item case $\EE = \EE_1+\EE_2$: if $e,e'\in E$ (or $\FF$) the we are done by induction hypothesis. Otherwise, $e,e'\in \in(\EE_1+\EE_2)$ and we are done because $x=\emptyset$.
\item case $\EE = \EE_1\cdot\EE_2$: then either $x\in\C(\EE_1)$ and $e\in E_1$ or $x = x_{\EE_1}\cup x_{\EE_2}$ for some $x_{\EE_1}\in\C(\EE_1)$ and $ x_{\EE_2}\in\C(\EE_2)$ and $e\in E_2$. The result follow by induction hypothesis.
\item case $\EE = \EE_1\|\EE_2$: we have $x = \{f\}\cup x_{\EE_1}\cup x_{\EE_2}$ for some $x_{\EE_1}\in\C(\EE_1)$, $ x_{\EE_2}\in\C(\EE_2)$ and $\{f\} = \init(\EE_1\|\EE_2)$. If $e\in E_1$ then $x_{\EE_1}\cup\{e\}\in\C(\EE_1)$ and the result follows by induction hypothesis. Similarly for $e\in E_2$. If $\{e\} = \exit_{\EE}$ then the result is trivial because $\<e\> = \{e\}$ i.e. $e=e'$.
\item case $\EE = \EE_1*\EE_2$: by construction, there exists $i\in\N$ such that $e$ is an event of $\EE_1*_{\le i}\EE_2$. Since $x\cup\{e\}$ is a configuration, $x$ is necessary a configuration of $\EE*_{\le i}\EE_2$(Corollary~\ref{pro:seq-maximal}). The result follows from the previous cases of $(+)$ and $(\cdot)$ and the induction hypothesis.\qed
\end{itemize}
\end{proof}

\begin{repproposition}{\ref{pro:precongruence}.}
$\refby$ is a precongruence i.e. for every pBES $(\EE,\pi), (\FF,\rho)$ and $(\G,\eta)$, if $(\EE,\pi)\refby (\FF,\rho)$ then $(\EE,\pi)\bullet(\G,\eta)\refby (\FF,\rho)\bullet(\G,\eta)$ (and symmetrically) for every $\bullet\in\{+,\cdot,\|,*)$.
\end{repproposition}

\begin{proof}
The case of $+$ is clear. For $\|$, Let $(\EE,\pi)\refby(\FF,\rho)$ be witnessed by a simulation $S\subseteq\C(\EE)\times\D(\C(\FF))$ and $(\G,\eta)$ be any pBES. We construct a relation $R\subseteq\C(\EE\|\G)\times\D(\C(\FF\|\G))$ such that $(x,\Theta)\in R$ iff $x = x_\EE\cup x_\G\cup z$ and $\Theta = \sum_i\alpha_i\delta_{y_i\cup x_\G\cup z}$ where $(x_\EE,\sum_i\alpha_i\delta_{y_i})\in S$, $x_\G\in\C(\G)$ and $z \in\{ \{e\},\{e,f\}\}$ and $e,f$ are the delimiters introduced by $\|$. Let us show that $R$ is indeed a simulation.

\begin{itemize}
\item That $(\emptyset,\delta_\emptyset)\in R$ is clear.
\item Let $(x,\Theta)\in R$ such that $x = x_\EE\cup x_\G\cup z$ and $\Theta = \sum_i\alpha_i\delta_{y_i\cup x_\G\cup z}$. Since $(x_\EE,\sum_i\alpha_i\delta_{y_i})\in S$, we have $x_\EE\refby_s y_i$ for every $i$. Therefore, $x\refby y_i\cup x_\G\cup z$ for every $i$.
\item Let $(x,\Theta)\in R$ such that $x = x_\EE\cup x_\G\cup z$ and $\Theta = \sum_i\alpha_i\delta_{y_i\cup x_\G\cup z}$. Let us write $\Theta_\FF = \sum_i\alpha_i\delta_{y_i}$. Assume that $x\prefix\Delta'$ in $(\EE\|\G,\pi\cup \eta)$, then either the prefix relation is obtained from $\pi$ or $\eta$.
\begin{itemize}
\item If the step is made in $\EE$ then 
$\Delta' = \sum_{e\in\supp(p)}p.e\delta_{x\cup\{e\}}$ for some $p\in\pi$. In particular, $x_\EE\prefix\sum_{e\in\supp(p)}p.e\delta_{x_\EE\cup\{e\}}$ in $\EE$ and since $(x_\EE,\Theta_\FF)\in S$, there exists $\Theta_\FF'$ such that $\Theta_\FF\ov\prefix^*\Theta_\FF'$ and $(\sum_{e\in\supp(p)}p.e\delta_{x_\EE\cup\{e\}},\Theta_\FF')\in\ov S$. By definition of lifting, there exists $\Theta_{\FF,e}\in\D(\C(\FF))$, for each $e\in\supp(p)$, such that $(x_\EE\cup\{e\},\Theta_{\FF,e})\in S$ and $\Theta_\FF' = \sum_{e\in\supp(p)}p.e\Theta_{\FF,e}$. If $\Theta_{\FF,e} = \sum_j\beta_j^e\delta_{y_j}$ then we consider $\Theta_{e} = \sum_j\beta_j^e\delta_{y_j\cup x_\G\cup z}$. By definition of $R$, we have $(x\cup\{e\},\Theta_{e})\in R$ for every $e'\in\supp(p)$ and Proposition~\ref{pro:lifting} implies that $(\Delta',\sum_{e'\in\supp(p)}p.e'\Theta_{e'})\in \ov R$ and it is clear that $\Theta\ov\prefix^*\sum_{e'\in\supp(p)}p.e'\Theta_{e'}$.

\item If the step is made in $\G$ then $\Delta' = \sum_{e\in\supp(p)}p.e\delta_{x_\EE\cup x_\G\cup\{e\}}$ for some $p\in\eta$. As before, 
$$x_\G\prefix\sum_{e\in\supp(p)}p.e\delta_{x_\G\cup\{e\}}$$ and therefore $\Theta\prefix\sum_i\sum_{e\in\supp(p)}\alpha_ip.e\delta_{y_i\cup x_\G\cup\{e\}\cup z} = \Theta'$ and $(\Delta',\Theta')\in\ov R$ can be deduced using Proposition~\ref{pro:lifting}.
\end{itemize}
\item It is obvious that $R$ preserves configuration final events.
\end{itemize}
The same simulation $R$ can be used to prove monotonicity of $\cdot$ because $\C(\EE\cdot\G)\subseteq\C(\EE\|\G)$. The only difference when $x\prefix $.

For $*$, we use Proposition~\ref{pro:kleene-star} (which is proved by direct simulation construction) and monotonicity of the other operators. Let $\EE\refby\FF$, we remove the local probability for simplicity. Since $\FF\cdot(\FF*\G)\refby\FF*\G$, monotonicity of sequential composition implies $\EE\cdot(\FF*\G)\refby\FF\cdot(\FF*\G)\refby\FF*\G$. But $\G\refby\FF*\G$, therefore we have $\EE\cdot\FF*\G + \G\refby\FF*\G$ and by Axiom~(\ref{eq:induction}), we have $\EE*\G\refby\FF*\G$.  The symmetric inequality follows from the left monotonicity of ($\cdot$). \qed
\end{proof}

\begin{repproposition}{\ref{pro:subdistributivity}}
The subdistributivity laws~(\ref{eq:+-subdist-seq},\ref{eq:pc-supdist-seq}) and (\ref{eq:+-subdist-par},\ref{eq:pc-supdist-par}) and the interchange law~(\ref{eq:interchange-law}) hold for regular pBES modulo probabilistic simulation.
\end{repproposition}

\begin{proof}
We give the complete proof for Equation~(\ref{eq:pc-supdist-par}) (Equation~(\ref{eq:pc-supdist-seq}) in a similar fashion) and~(\ref{eq:interchange-law}). Two copies of $\EE$ are made for the distributed expression and the fresh events introduced by each $\|$ are respectively denoted by $e,e_1,e_2,f,f_1,f_2$. We construct a relation 
$$S\subseteq\C(\EE\|(\FF+\G))\times\D(\C(\EE_1\|\FF+\EE_2\|\G))$$ such that $(x,\Theta)\in S$ if one of the following cases hold:
\begin{itemize}
\item $x\in\C(e\cdot \EE)$ and $\Theta = (1-\alpha)\delta_{x_1} + \alpha\delta_{x_2}$,
\item $x = z\cup x_\EE\cup x_\FF$ such that $x_\EE\in\C(\EE),x_\FF\in\C(\FF)\setminus\{\emptyset\}$ and $\Theta = \delta_{x_{\EE_1}\cup x_\FF\cup z_1}$,
\item $x = z\cup x_\EE\cup x_\G$ such that $x_\EE\in\C(\EE),x_\G\in\C(\G)\setminus\{\emptyset\}$ and $\Theta = \delta_{x_{\EE_2}\cup x_\G\cup z_2}$,
\end{itemize}
where $z\in\{\{e\},\{e,f\}\}$ (resp. for $z_i$). We show that $S$ is indeed a probabilistic simulation.
\begin{itemize}
\item It is clear that $(\emptyset,\delta_\emptyset)\in S$ because $\emptyset\in\C(\EE)$ and $\delta_\emptyset = (1-\alpha)\delta_\emptyset + \alpha\delta_\emptyset$.
\item Let $(x,\Theta)\in S$, in all three cases, we have $x\refby_s y$ for all $y\in\supp(\Theta)$.
\item Let $(x,\Theta)\in S$ and $x\prefix\Delta'$. By definition of $+$, $\|$ and $\prefix$, there are four cases.
\begin{itemize}
\item $\supp(\Delta')\subseteq\C(\EE)$: since $x\subseteq y$ for all $y\in\supp(\Delta')$, we have $x\subseteq E$. Therefore, $\Theta = (1-\alpha)\delta_{x_1}+\alpha\delta_{x_2}$ and therefore  $x_i\prefix\Delta'_i$ (the copies of $\Delta'$). Therefore, Proposition~\ref{pro:lifting} implies that $\Theta\prefix (1-\alpha)\Delta_1'+\alpha\Delta_2'$ and $(\Delta',(1-\alpha)\Delta_1'+\alpha\Delta_2')\in\ov S$.
\item $\supp(\Delta')\subseteq\C(\FF)$: this implies that $x_\FF \neq \emptyset$ then the result is clear. 
\item $\supp(\Delta')\subseteq\C(\G)$: as above.
\item $\supp(\Delta') = (1-\alpha)\sum_{e\in\init(\FF)}p.e\delta_{x\cup e} + \alpha\sum_{e\in\init(\G)}q.e\delta_{x\cup e}$ by definition of $\FF\pc{\alpha}\G$. By definition of $R$, $\Theta = (1-\alpha)\delta_{x_1} + \alpha\delta_{x_2}$ and since $x_1\cup\{e\}$ is a configuration for every $e\in\init(\FF)$ (reps. for $\G$), we have $\Theta\ov\prefix\Theta' = (1-\alpha)\sum_{e}p.e\delta_{x_1\cup \{e\}} + \alpha\sum_{e}q.e\delta_{x_2\cup\{e\}}$ and $\Delta'\ov R\Theta'$ by Definition of lifting.
\end{itemize}
\item Since the final events are respectively $\{f\}$ and $\{f_1,f_2\}$ for the left and right hand side, it is clear that $R$ satisfies the last property of a simulation.
\end{itemize}
For Equation~(\ref{eq:interchange-law}), let us write $e,f$ and $e',f'$ the respective events introduced as delimiters in $\EE\|\FF$ and $\EE'\|\FF'$. The delimiters of $(\EE\cdot\EE')\|(\FF\cdot\FF')$ are $e,f'$. We consider the relation $S\subseteq \C(\EE\|\FF)\cdot(\EE'\|\FF')\times (\EE\cdot\EE')\|(\FF\cdot\FF')$ such that $(x,y)\in S$ iff $y = x\setminus\{f,e'\}$. It then follows easily that the probabilistic relation $(x,\delta_y)\in R$ iff $(x,y)\in S$ is indeed a simulation.
\end{proof}

\begin{repproposition}{\ref{pro:kleene-star}}
The Kleene star satisfies equation~(\ref{eq:unfold}) and the equational  implication~(\ref{eq:induction}).
\end{repproposition}
\begin{proof}
Let $S\subseteq\C(\EE\cdot\FF + \G)\times\D(\C(\FF))$ be a probabilistic simulation from $(\FF,\rho)\cdot(\EE,\pi) + (\G,\eta)$ to $(\FF,\rho)$. Again, we will leave the set of distributions $\pi,\rho$ and $\eta$ implicit in this proof. By hypothesis, $\EE\cdot\FF + \G\refby \FF$ and $\EE*_0\G = \G\refby\FF$,  $\EE*_1\G = \G + \EE\cdot\G\refby\FF$ and simple induction shows that $\EE*_{\le i}\G\refby\FF$ and we denote such simulation by $S^{(i)}$. Moreover, since $\EE*_{\le i-1}\G\prefix\EE*_{\le i}\G$, we can find a family of simulations such that $S^{(i-1)}$ is the restriction of $S^{(i)}$ to $\EE*_{\le i-1}\G$. Therefore, we consider the reunion $S^* = \cup_i S^{(i)}$ and show that it is indeed a simulation from $\EE*\G$ to $\FF$.
\begin{itemize}
\item It is clear that $(\emptyset,\delta_\emptyset)\in S$.
\item Let $(x,\Theta)\in S$. Since $x$ and every $y$ configurations in $\supp(\Theta)$ are finite configurations, $(x,\Theta)\in S^{(i)}$ for some $i$ and we deduce that $x\refby_s y$.
\item Let $(x,\Theta)\in S$ and $x\prefix\Delta'$. Let $i$ be some integer such that $(x,\Theta)\in S^{(i)}$. There are two cases:
\begin{itemize}
\item if $\supp(\Delta')\subseteq \C(\EE*_{\le i} \G)$ then we are done because $S^{(i)}$ is a probabilistic simulation.
\item if $\supp(\Delta')\subseteq \C(\EE*_{\le i+1}\G)$ then we have $(x,\Theta)\in S^{(i+1)}$ because $S^{(i)}$ is the restriction of $S^{(i+1)}$ and we are done because $S^{(i+1)}$ is a simulation. 
\end{itemize}
\item Let $(x,\Theta)\in S$ and $x\cap\exit_{\EE*\G}$. Since $x$ is a finite set, $x\in \C(\EE*_{\le i}\G)$ for some $i$ and the result follows form $S^{(i)}$.\qed
\end{itemize}
\end{proof}

\end{document}